\documentclass[aps,pre,preprint,superscriptaddress,showpacs]{revtex4-1}

\usepackage{latexsym,amsmath,amssymb,amsthm,bm,booktabs,enumerate,url}
\usepackage[dvipdfmx]{graphicx}

\newtheorem{thm}{Theorem}
\newtheorem{lem}{Lemma}

\newtheorem{problem}{Problem}

\begin{document}


\title{Maximizing Barber's bipartite modularity is also hard}


\author{Atsushi Miyauchi}
\email[]{miyauchi.a.aa@m.titech.ac.jp}
\affiliation{Graduate School of Decision Science and Technology, Tokyo Institute of Technology,  2-12-1 Ookayama, Meguro-ku, Tokyo 152-8552, Japan}
\affiliation{JST, ERATO, Kawarabayashi Large Graph Project, c/o Global Research Center for Big Data Mathematics,
NII, 2-1-2 Hitotsubashi, Chiyoda-ku, Tokyo 101-8430, Japan}
\author{Noriyoshi Sukegawa}
\email[]{sukegawa.n.aa@m.titech.ac.jp}
\affiliation{Graduate School of Decision Science and Technology, Tokyo Institute of Technology, 2-12-1 Ookayama, Meguro-ku, Tokyo 152-8552, Japan}


\date{\today}

\begin{abstract}
Modularity introduced by Newman and Girvan~[Phys. Rev. E \textbf{69}, 026113 (2004)] is a quality function for community detection. Numerous methods for modularity maximization have been developed so far. In 2007, Barber~[Phys. Rev. E \textbf{76}, 066102 (2007)] introduced a variant of modularity called bipartite modularity which is appropriate for bipartite networks. Although maximizing the standard modularity is known to be NP-hard, the computational complexity of maximizing bipartite modularity has yet to be revealed. In this study, we prove that maximizing bipartite modularity is also NP-hard. More specifically, we show the NP-completeness of its decision version by constructing a reduction from a classical partitioning problem. 
\end{abstract}


\maketitle

\section{Introduction} \label{sec: introduction}
Networks have attracted much attention from diverse fields such as physics, informatics, chemistry, biology, sociology, and so forth. Many complex systems arising in such fields can be represented as networks, and analyzing the structures and dynamics of these networks provides meaningful information about the underlying systems~\cite{Ne03, Ne09}.

In network analysis, one of the most fundamental issues now is finding community structures. Roughly speaking, communities are the sets of vertices densely connected inside, but sparsely connected with the rest of the network. Community detection analysis is increasingly applied in various fields. For more details, see the useful survey by Fortunato~\cite{Fo10} with over 450 references.

Community detection is now often conducted through maximizing a quality function called \textit{modularity} introduced by Newman and Girvan~\cite{NeGi04}. This function was solely a quality measure at first, but nowadays it is widely used as an objective function of optimization problems for finding community structures. Modularity represents the sum, over all communities, of the fraction of the number of the edges connecting vertices in a community minus the expected fraction of the number of such edges assuming that they are put at random with the same distribution of vertex degree. Let us consider an undirected network $G=(V, E)$ consisting of $n=|V|$ vertices and $m = |E|$ edges, and take a division $\mathcal{C}$ of $V$, then modularity $Q$ can be written as:
\begin{equation*} \label{def: modularity}
Q(\mathcal{C}) = \sum_{C \in \mathcal{C}} \left(\frac{m_C}{m} - \left(\frac{D_C}{2m}\right)^2 \right), 
\end{equation*}
where $m_C$ is the number of all the edges connecting vertices in community $C$, and $D_C$ is the sum of the degrees of all the vertices in community $C$. In 2008, Brandes \textit{et al.}~\cite{BrDeGaGoHoNiWa08} provided the first computational complexity result for modularity maximization. More precisely, they showed that modularity maximization is NP-hard. In other words, unless $\text{P} = \text{NP}$, there exists no modularity maximization algorithm that simultaneously satisfies the following: (i) finds a division with maximal modularity (ii) in time polynomial in $n$ and $m$ (iii) for any networks. Numerous heuristics based on greedy techniques~\cite{NeGi04, ClNeMo04, BlGuLaLe08}, simulated annealing~\cite{GuAm05, MaDo05, MeAcDo05}, spectral optimization~\cite{Ne06}, extremal optimization~\cite{DuAr05}, dynamical clustering~\cite{BoIvLaPlRa07}, mathematical programming~\cite{AgKe08, CaHaLi11, MiMi13}, and other techniques have been developed. In addition, a few exact algorithms~\cite{XuTsPa07, AlCaCaHaPeLi10} have also been proposed.

In recent years, some authors have reported that modularity is not perfect because it has two drawbacks: the resolution limit~\cite{FoBa07} and degeneracies~\cite{GoMoCl10}. The former means that, when the number of edges is large, small communities tend to be put together even if they are cliques connected by only one edge. The latter means that there exist a large number of nearly optimal divisions in terms of modularity maximization, which makes finding communities with maximal modularity extremely difficult. Nevertheless, modularity maximization is regarded as the most popular approach for community detection.

There are some variants of modularity such as ones for multi-scale community detection~\cite{ReBo06, ArFeGo08, HeKuKaSa08, PoLa11} and ones applicable to weighted or directed networks~\cite{Ne04, ArDuFeGo07, LeNe08, KiSoJe10}.  In 2007, Barber~\cite{Ba07} proposed one of such variants called \textit{bipartite modularity} for community detection in bipartite networks. Needless to say, the standard modularity is applicable to bipartite networks. However, it does not reflect a structure and restrictions of bipartite networks, that is, the vertices in a bipartite network can be divided into two disjoint sets of red and blue vertices such that every edge connects a red vertex and a blue vertex. 
Barber's bipartite modularity $Q_b$ does reflect them, and it can be represented as: 
\begin{equation*} \label{def: bipartite modularity}
Q_b(\mathcal{C}) = \sum_{C \in \mathcal{C}} \left(\frac{m_C}{m} - \frac{R_CB_C}{m^2} \right), 
\end{equation*}
where $R_C$ is the sum of the degrees of all the red vertices in community $C$, and $B_C$ is the same for the blue vertices. It can be seen that the terms $2m$ and $D_C^2$ in the standard modularity are replaced by $m$ and $R_CB_C$, respectively. These modifications are due to a structure and restrictions of bipartite networks. As the standard modularity, many approaches have been proposed so far~\cite{Ba07, ZhZhGuZh11, CoHa13} because bipartite networks arise in various real-world systems. We note that there is another variant for bipartite networks proposed by Guimer\`a, Sales-Pardo, and Amaral~\cite{GuPaAm07}, which is also often employed. 


As mentioned above, maximizing the standard modularity is known to be NP-hard. On the other hand, the computational complexity of maximizing bipartite modularity has yet to be revealed. In 2011, Zhan \textit{et al}.~\cite{ZhZhGuZh11} stated that maximizing the standard modularity can be reduced to maximizing bipartite modularity. If this is correct, then we can conclude that maximizing bipartite modularity is NP-hard. However, as pointed out by Costa and Hansen~\cite{CoHa11}, their analysis includes a crucial error. In 2013, Costa and Hansen~\cite{CoHa13} stated that the computational complexity of maximizing bipartite modularity still remains open. 

In this study, we prove that maximizing bipartite modularity is also NP-hard. To this end, we show the NP-completeness of its decision version by constructing a reduction from a classical partitioning problem. We note that our analysis is based on that of Brandes \textit{et al}.~\cite{BrDeGaGoHoNiWa08} who succeeded to show the NP-hardness of maximizing the standard modularity. 


\section{NP-completeness} \label{sec: np-completeness}
In what follows, we study the following problem which is the decision version of maximizing bipartite modularity.
\begin{problem}[BIMODULARITY]
Given a bipartite network $G=(V, E)$ and a real number K, does there exist a division $\mathcal{C}$ of $V$ such that $Q_b(\mathcal{C}) \geq K$?
\end{problem}

Our analysis employs the following partitioning problem as Brandes \textit{et al.}~\cite{BrDeGaGoHoNiWa08} did for the standard modularity.

\begin{problem}[3-PARTITION]
Given a set of $3k$ positive integers $A = \{a_1, a_2, \dots, a_{3k}\}$ such that $a = \sum_{i=1}^{3k} a_i = kb$ and $b/4 < a_i < b/2$ for $i=1, 2, \dots, 3k$, for some integer $b$, does there exist a partition of $A$ into $k$ sets such that the sum of the numbers in each set is equal to $b$?
\end{problem}
3-PARTITION is NP-complete in the strong sense~\cite{GaJo79}, which means that the problem cannot be solved even in pseudo-polynomial time, unless $\text{P} = \text{NP}$. Therefore, to show the NP-completeness of BIMODULARITY, it is enough to construct a pseudo-polynomial time reduction from 3-PARTITION. In other words, we need to show that a given instance $A$ of 3-PARTITION can be transformed into a certain instance $(G(A), K(A))$ of BIMODULARITY such that $G(A)$ has a division $\mathcal{C}$ of $V$ which satisfies $Q_b(\mathcal{C}) \geq K(A)$ if and only if $A$ can be partitioned into $k$ sets with sum equal to $b$ each.

We initially propose a procedure for generating appropriate bipartite network $G(A)$ from $A = \{a_1, a_2, \dots, a_{3k}\}$ as follows:
\begin{description}
\item[Step 1] Construct $k$ complete bipartite networks (\textit{bicliques} for short) $K_1, K_2, \dots, K_k$ consisting of $a$ red vertices and $a$ blue vertices. 
\item[Step 2] For each $a_i \in A$, put a red vertex $x_i$ and a blue vertex $y_i$. These are termed \textit{element vertices}. 
\item[Step 3] For $i = 1, 2, \dots, 3k$, connect $x_i$ to $a_i$ blue vertices in each of $k$ bicliques constructed in Step~1 such that each blue vertex in bicliques is connected to exactly one red element vertex. For $i = 1, 2, \dots, 3k$, connect $y_i$ to $a_i$ red vertices in each of $k$ bicliques constructed in Step~1 in a similar manner.
\item[Step 4] For $i = 1, 2, \dots, 3k$, connect the pair of element vertices $x_i$ and $y_i$.
\item[Step 5] For $i = 1, 2, \dots, 3k$, construct a star $X_i$ consisting of one blue internal vertex and $a^2/7$ red leaves, and construct a star $Y_i$ consisting of one red internal vertex and $a^2/7$ blue leaves. (Note that we can assume that $a^2$ is a multiple of 7 because all instances of 3-PARTITION can be transformed into one that satisfies it.) 
\item[Step 6] For $i = 1, 2, \dots, 3k$, connect $x_i$ to the internal vertex of $X_i$, and connect $y_i$ to the internal vertex of $Y_i$.
\end{description}

This procedure generates a bipartite network $G(A)$ consisting of 
\begin{equation*}
n = \frac{6}{7}ka^2 + 2ka + 12k
\end{equation*}
vertices and 
\begin{equation*}
m = \frac{13}{7}ka^2 + 2ka + 9k
\end{equation*}
edges. Clearly, it can be done in pseudo-polynomial time, that is, polynomial time in the sum of the input values of $A$. We note that each vertex of the bicluques $K_1, K_2, \dots, K_k$ has degree $a + 1$, and for each $a_i \in A$ the element vertices $x_i$ and $y_i$ have degrees $ka_i + 2$. In Fig.~\ref{fig: gadget}, $G(A)$ constructed from $A = \{2,2,2,2,3,3\}$ is shown as an example. 
\begin{figure}
\includegraphics[width=14.5cm]{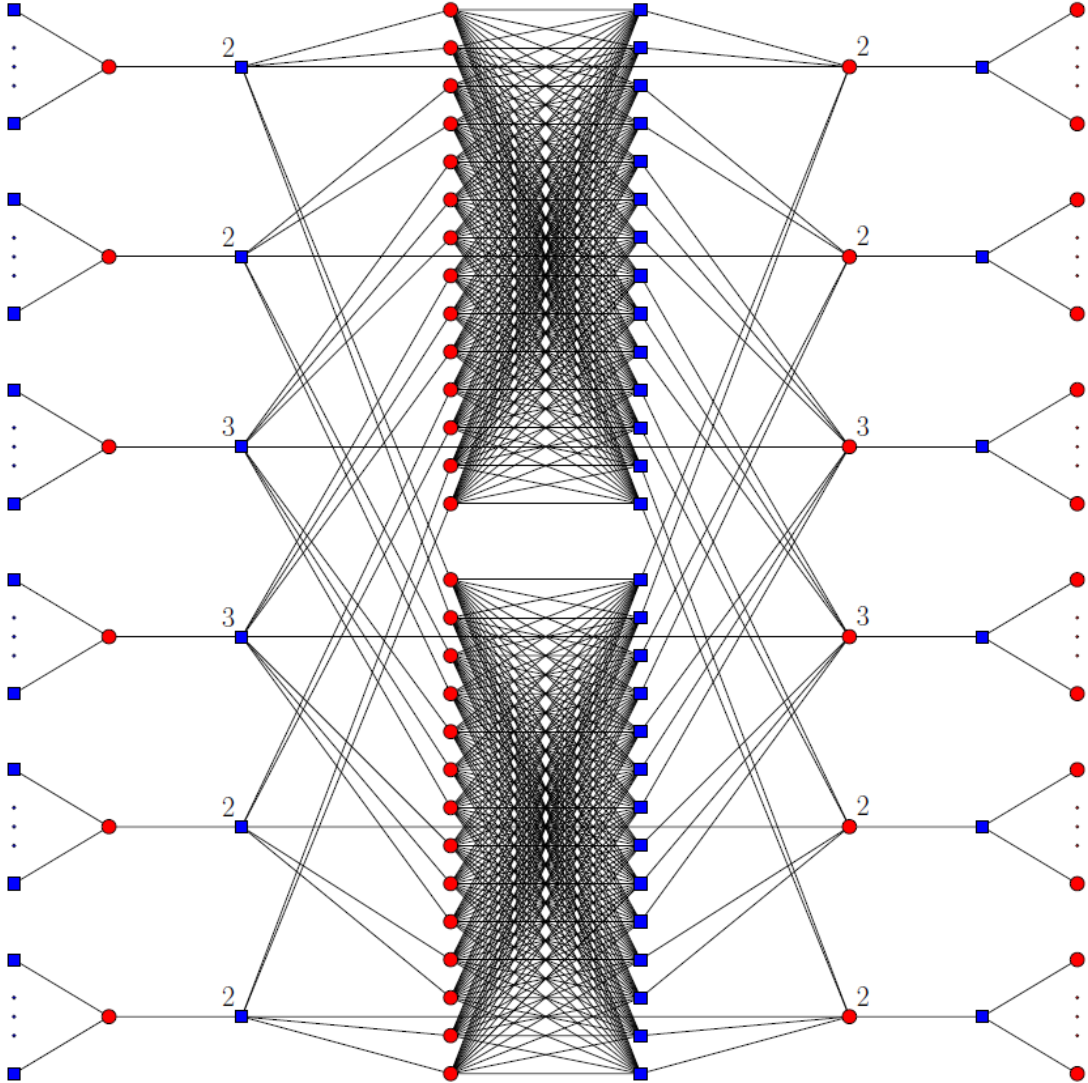}
\caption{(Color online) Bipartite network $G(A)$ constructed from $A = \{2,2,2,2,3,3\}$. Labels of vertices represent corresponding elements in $A$.}
\label{fig: gadget}
\end{figure}

Before determining appropriate parameter $K(A)$ for the instance of BIMODULARITY, we observe several conditions satisfied by divisions of  $G(A)$ with maximal bipartite modularity.

\begin{lem} \label{lem: 1}
In any division of $G(A)$ with maximal bipartite modularity, none of the bicliques $K_1, K_2, \dots, K_k$ is divided.
\end{lem}

\begin{proof}
Let us consider an arbitrary division $\mathcal{C}$. Suppose that a biclique $K_t$ is divided into $l$ communities with $l > 1$ in division $\mathcal{C}$. We denote the communities containing vertices of $K_t$ by $C_1, C_2, \dots, C_l$. The contribution of $C_1, C_2, \dots, C_l$ to $Q_b$ can be written as:
\begin{equation*}
\frac{1}{m} \sum_{i=1}^l m_i - \frac{1}{m^2} \sum_{i=1}^l R_i B_i,
\end{equation*}
where $m_i$ is the number of the edges connecting vertices in $C_i$, $R_i$ is the sum of the degrees of the red vertices in $C_i$, and $B_i$ is the same for the blue vertices.

Transform $C_1, C_2, \dots, C_l$ into $C'_1, C'_2, \dots, C'_l$ by removing all the vertices of $K_t$ from each community. We construct a new division $\mathcal{C'}$ by replacing $C_1, C_2, \dots, C_l$ in $\mathcal{C}$ with $K_t, C'_1, C'_2, \dots, C'_l$. For $i = 1, 2, \dots, l$, we denote the number of the red vertices removed from $C_i$ by $r_i$, the same for the blue vertices by $b_i$, and the number of the edges between vertices of $K_t$ in $C_i$ and the element vertices in $C_i$ by $f_i$. Then, the decrement, due to the transformation from $\mathcal{C}$ into $\mathcal{C}'$, of the number of the edges within communities is given by $\sum_{i=1}^l f_i$. On the other hand, the increment of the number of such edges can be represented as $\sum_{i=1}^l \sum_{j \neq i}^l r_i b_j$ because biclique $K_t$ is added to $\mathcal{C}'$ as a new community. Additionally, as for the sum of the degrees of $C_i$, the red one decreases $(a + 1)r_i$ and the blue one decreases $(a + 1)b_i$ because each vertex of $K_t$ has degree $a + 1$. From the above, the contribution of $K_t, C'_1, C'_2, \dots, C'_l$ to $Q_b$ is calculated by 
\begin{eqnarray*}
\begin{split}
&\frac{1}{m} \left(\sum_{i=1}^l m_i - \sum_{i=1}^l f_i + \sum_{i=1}^l \sum_{j \neq i}^l r_i b_j \right) \\
&-\frac{1}{m^2} \left((a + 1)^2 a^2 + \sum_{i=1}^l \left(R_i - (a + 1)r_i \right) \left(B_i - (a + 1)b_i \right) \right).
\end{split}
\end{eqnarray*}
Thus, we see that 
\begin{eqnarray*}
\Delta&:=&Q_b(\mathcal{C}') - Q_b(\mathcal{C}) \\
&=&\frac{1}{m} \left(\sum_{i=1}^l \sum_{j \neq i}^l r_i b_j - \sum_{i=1}^l f_i \right) \\
&&+ \frac{1}{m^2} \left( (a + 1) \left(\sum_{i=1}^l R_i b_i + \sum_{i=1}^l B_i r_i - (a + 1) \sum_{i=1}^l r_i b_i \right) - (a + 1)^2 a^2 \right).
\end{eqnarray*}

Now, for $i = 1, 2, \dots, l$, decompose the set of the edges enumerated by $f_i$ into the set of $f_{i1}$ edges and the set of $f_{i2}$ edges which are incident to red vertices and blue vertices of $K_t$, respectively. Then, we obtain that $R_i \geq (a + 1)r_i + kf_{i2}$ for $i=1,2,\dots,l$. This inequality holds because $C_i$ has at least $r_i$ red vertices of $K_t$ with degree $a + 1$. Moreover, it also contains some red element vertices which have at least $f_{i2}$ edges connected to biclique $K_t$, and such element vertices again have at least $f_{i2}$ edges to each of the other $k - 1$ bicliques. Arguing similarly for blue vertices, we also obtain that $B_i \geq (a + 1)b_i + kf_{i1}$ for $i=1,2,\dots,l$. Therefore, it holds that
\begin{equation*}
\sum_{i=1}^l R_i b_i \geq \sum_{i=1}^l \left((a + 1)r_i + kf_{i2} \right) b_i = (a + 1) \sum_{i=1}^l r_i b_i + k \sum_{i=1}^l f_{i2} b_i, 
\end{equation*}
and 
\begin{equation*}
\sum_{i=1}^l B_i r_i \geq \sum_{i=1}^l \left((a + 1)b_i + kf_{i1} \right) r_i = (a + 1) \sum_{i=1}^l r_i b_i + k \sum_{i=1}^l f_{i1} r_i.
\end{equation*}
Using these inequalities and the following equalities
\begin{equation*}
\sum_{i=1}^l \sum_{j \neq i}^l r_i b_j = \sum_{i=1}^l \sum_{j=1}^l r_i b_j - \sum_{i=1}^l r_i b_i = a^2 - \sum_{i=1}^l r_i b_i, 
\end{equation*}
we see that 
\begin{eqnarray*}
\Delta &\geq& \frac{1}{m} \left(a^2 - \sum_{i=1}^l r_i b_i - \sum_{i=1}^l f_i \right)\\
	&& + \frac{1}{m^2} \left( (a + 1)^2\sum_{i=1}^l r_i b_i + k(a + 1)\sum_{i=1}^l (f_{i1} r_i + f_{i2} b_i) - (a + 1)^2a^2\right)\\
&=&\frac{1}{m^2} \left( ma^2 - (a + 1)^2 a^2 - (m - (a + 1)^2) \sum_{i=1}^l r_i b_i \right. \\ 
&&\left. - m \sum_{i=1}^l f_i + k(a + 1) \sum_{i=1}^l (f_{i1} r_i + f_{i2} b_i) \right).
\end{eqnarray*}
Now, focusing on the last two terms in the above parenthesis, we see that  
\begin{eqnarray*}
&&-m \sum_{i=1}^l f_i + k(a + 1) \sum_{i=1}^l (f_{i1} r_i + f_{i2} b_i) \\
&=&\sum_{i=1}^l f_{i1} \left(k(a + 1) r_i - m \right) + \sum_{i=1}^l f_{i2} \left(k(a + 1) b_i - m \right) \\
&\geq&\sum_{i=1}^l r_i \left(k(a + 1) r_i - m \right) + \sum_{i=1}^l b_i \left(k(a + 1) b_i - m \right) \\
&=&k(a + 1) \sum_{i=1}^l \left(r_i^2 + b_i^2 \right) - 2ma \\
&\geq&2k(a+1)\sum_{i=1}^l r_ib_i - 2ma.
\end{eqnarray*}
The first equality follows because $f_i = f_{i1} + f_{i2}$ for $i=1, 2, \dots, l$. The next inequality follows because we have $r_i \geq f_{i1}$ and $b_i \geq f_{i2}$ for $i=1, 2, \dots, l$. Therefore, it holds that 
\begin{eqnarray*}
\Delta &\geq& \frac{1}{m^2} \left( ma^2 - (a+1)^2a^2 - 2ma -\left(m - (a+1)^2 -2k(a+1) \right) \sum_{i=1}^l r_ib_i \right) \\
&\geq& \frac{a}{m^2} \left( ma - (a+1)^2a - 2m -\left(m - (a+1)^2 -2k(a+1) \right) (a - 1) \right) \\
&=&\frac{a}{m^2} \left(2k(a+1)a - m - (a + 1)^2 -2k(a+1)\right) \\
&=&\frac{a}{m^2} \left( \left(\frac{1}{7} a^2 - 2a - 2\right) k - a^2 -2a -1 \right).
\end{eqnarray*}
The second inequality holds because $m - (a + 1)^2 -2k(a + 1) > 0$ as $k \geq 1$ and $\sum_{i=1}^l r_ib_i \leq a(a-1)$ as biclique $K_t$ was divided into at least two communities in $\mathcal{C}$. We can assume that $k$ is greater than any constant for all relevant instances of 3-PARTITION. Thus, employing $k > 14$ for an example, we see that 
\begin{equation*}
\Delta > \frac{a}{m^2}\left(a^2 - 30a - 29\right) > 0, 
\end{equation*}
as $a \geq 3k > 42$. Taking an optimal division as $\mathcal{C}$ at first, we obtain a contradiction. 
\end{proof}

\begin{lem} \label{lem: 2}
In any division of $G(A)$ with maximal bipartite modularity, every community contains at most one of the bicliques $K_1, K_2, \dots, K_k$.
\end{lem}


\begin{proof}
Let us consider an arbitrary optimal division $\mathcal{C}^*$. From Lem.~\ref{lem: 1}, none of the bicliques $K_1, K_2, \dots, K_k$ is divided. Suppose that $\mathcal{C}^*$ has a community $C$ which contains $l$ of the bicliques with $l \geq 2$. The set of indices of the red element vertices in $C$ is denoted by $I_R$, and the same for the blue vertices is denoted by $I_B$. Additionally, the sum of the degrees of the red vertices other than the $l$ bicliques is denoted by $R$, and the same for the blue vertices is denoted by $B$. Then, the contribution of $C$ to $Q_b$ is calculated by 
\begin{equation*}
\frac{1}{m} \left(a^2 l + l \left(\sum_{i \in I_R} a_i + \sum_{i \in I_B} a_i \right) + |I_R \cap I_B| \right) - \frac{1}{m^2} \left( (a + 1) al + R \right) \left( (a + 1) al + B \right). 
\end{equation*}
Note that $|I_R \cap I_B|$ enumerates the number of the edges between the element vertices $x_i$ and $y_i$ in $C$.

Let us take an arbitrary biclique $K_t$ from the bicliques contained in $C$. Construct a new division $\mathcal{C}'$ by dividing $C$ into $K_t$ and the rest $C'$. Clearly, the increment, due to this transformation, of the number of the edges within communities is 0. On the other hand, the decrement  is given by $\sum_{i \in I_R} a_i + \sum_{i \in I_B} a_i$ because all the edges between $K_t$ and the element vertices in $C$ are cut. Thus, the contribution of $K_t$ and $C'$ to $Q_b$ is calculated by 
\begin{eqnarray*}
\begin{split}
&\frac{1}{m} \left(a^2 l + (l - 1) \left(\sum_{i \in I_R} a_i + \sum_{i \in I_B} a_i \right) + |I_R \cap I_B| \right) \\
&-\frac{1}{m^2}\left((a + 1)^2 a^2 + \left((a + 1) a (l - 1) + R \right) \left((a + 1) a (l - 1) + B\right) \right).
\end{split}
\end{eqnarray*}
Hence, we see that 
\begin{eqnarray*}
\Delta&:=&Q_b(\mathcal{C}') - Q_b(\mathcal{C}^*) \\
&=&-\frac{1}{m} \left(\sum_{i \in I_R} a_i + \sum_{i \in I_B} a_i \right) + \frac{1}{m^2}\left(2(a + 1)^2 a^2 (l - 1) + (a + 1) a (R+B) \right) \\
&=&\frac{1}{m^2}\left(2(a + 1)^2 a^2 (l - 1) + (a + 1)a(R+B) - m \left(\sum_{i \in I_R} a_i + \sum_{i \in I_B} a_i \right) \right) \\
&\geq&\frac{2a}{m^2}\left((a + 1)^2 a - m \right) \\
&\geq&\frac{2ka}{m^2}\left(\frac{8}{7}a^2 + 12k - 6 \right) \\
&>&0.
\end{eqnarray*}
The first inequality follows because we have $l \geq 2$, $R + B \geq 0$, and $\sum_{i \in I_R} a_i + \sum_{i \in I_B} a_i \leq 2a$. This contradicts the optimality of division $\mathcal{C}^*$. 
\end{proof}

\begin{lem} \label{lem: 3}
In any division of $G(A)$ with maximal bipartite modularity, none of the stars $X_1, X_2, \dots, X_{3k}$ and $Y_1, Y_2, \dots, Y_{3k}$ is divided. 
\end{lem}

\begin{proof}
Let us consider an arbitrary optimal division $\mathcal{C}^*$, and an arbitrary leaf $l$. It is sufficient to show that $l$ is contained in the community to which its adjacent vertex belongs. Suppose otherwise, that is, $l$ belongs to a community $C_1$ and its adjacent vertex belongs to another community $C_2$. Now, assume that $l$ is a red leaf. (Note that the following discussion is also applicable to every blue leaf.) Construct a new division $\mathcal{C}'$ by transferring $l$ from $C_1$ to $C_2$. The increment of $Q_b$ is $1/m$ due to the edge between $l$ and its adjacent vertex. On the other hand, the decrement of $Q_b$ is strictly less than $1/m$ because the degree of $l$ is 1 and the sum of the degrees of the blue vertices in $C_2$ is less than $m$ as the set of the blue vertices in the whole network is divided into at least $k$ communities by Lem.~\ref{lem: 2}. Thus, we obtain that $Q_b(\mathcal{C}') > Q_b(\mathcal{C}^*)$, which contradicts the optimality of division $\mathcal{C}^*$.
\end{proof}

\begin{lem} \label{lem: 4}
In any division of $G(A)$ with maximal bipartite modularity,  every red element vertex $x_i$ and the adjacent star $X_i$ are not contained in the same community. A Similar statement holds for every blue element vertex $y_i$ and the adjacent star $Y_i$.
\end{lem}

\begin{proof}
Let us consider an arbitrary optimal division $\mathcal{C}^*$, and an arbitrary red element vertex $x_i$. Suppose that $\mathcal{C}^*$ has a community $C$ which contains both $x_i$ and the adjacent star $X_i$. From Lem.~\ref{lem: 3}, $X_i$ is entirely contained in $C$. Now, the sum of the degrees of the red vertices in $C$ other than $X_i$ is denoted by $R$, and the same for the blue vertices is denoted by $B$. Then, the contribution of $C$ to $Q_b$ is given by 
\begin{equation*}
\frac{m_C}{m} - \frac{1}{m^2} \left(R + \frac{a^2}{7} \right) \left(B + \left(\frac{a^2}{7} + 1 \right) \right), 
\end{equation*}
because $X_i$ consists of one blue vertex with degree $a^2/7 + 1$ and $a^2/7$ red leaves. 

Construct a new division $\mathcal{C}'$ by dividing $C$ into $X_i$ and the rest $C'$. Clearly, the contribution of $X_i$ and $C'$ to $Q_b$ is given by 
\begin{equation*}
\frac{m_C - 1}{m} - \frac{1}{m^2} \left(RB + \frac{a^2}{7} \left(\frac{a^2}{7} + 1 \right) \right).
\end{equation*}
Hence, we see that 
\begin{eqnarray*}
\Delta &:=& Q_b(\mathcal{C}') - Q_b(\mathcal{C}^*) \\
&=& -\frac{1}{m} + \frac{1}{m^2} \left(\left(\frac{a^2}{7} + 1 \right) R + \frac{a^2}{7} B \right).
\end{eqnarray*}
Since $C$ at least contains element vertex $x_i$ other than $X_i$, we obtain that $R \geq k a_i + 2$. Using this inequality and $B \geq 0$, we see that 
\begin{eqnarray*}
\Delta&\geq&-\frac{1}{m} + \frac{1}{m^2} \left(\frac{a^2}{7} + 1 \right) (k a_i + 2).
\end{eqnarray*}
We can assume that $a_i$ is greater than any constant for all relevant instances of 3-PARTITION. Thus, employing $a_i > 21$ for an example, we immediately obtain that $\Delta > 0$. This contradicts the optimality of division $\mathcal{C}^*$. It is easy to see that the above discussion is applicable to every blue element vertex $y_i$ and the adjacent star $Y_i$.
\end{proof}

\begin{lem} \label{lem: 5}
In any division of $G(A)$ with maximal bipartite modularity, every star $X_1, X_2, \dots, X_{3k}$ and $Y_1, Y_2, \dots, Y_{3k}$ itself forms a community.
\end{lem}

\begin{proof}
Let us consider an arbitrary optimal division $\mathcal{C}^*$, and an arbitrary star $X_i$. From Lem.~\ref{lem: 3}, $X_i$ is entirely contained in a community $C$. Therefore, it is sufficient to show that $C$ contains no vertices other than $X_i$. Suppose otherwise, that is, some vertex other than $X_i$ belongs to $C$. Since the only adjacent vertex $x_i$ of $X_i$ is not contained in $C$ from Lem.~\ref{lem: 4}, $X_i$ is not connected with the other vertices in $C$. Constructing a new division $\mathcal{C}'$ by dividing $C$ into $X_i$ and the rest, we obtain that $Q_b(\mathcal{C}') > Q_b(\mathcal{C}^*)$. This contradicts the optimality of $\mathcal{C}^*$. The above discussion also holds for an arbitrary star $Y_i$.
\end{proof}

\begin{lem} \label{lem: 6}
In any division of $G(A)$ with maximal bipartite modularity, every element vertex belongs to one of the communities corresponding to the bicliques $K_1, K_2, \dots, K_k$.
\end{lem}

\begin{proof}
Let us consider an arbitrary optimal division $\mathcal{C}^*$. From Lem.~\ref{lem: 5}, it is sufficient to show that there exists no community consisting of element vertices only. Suppose otherwise, that is, $\mathcal{C}^*$ has a community $C$ which consists of element vertices only. In what follows, we consider the following two cases: when $C$ contains both red and blue element vertices, and when $C$ contains either red or blue element vertices only.

First, we analyze the former case. If there exists a vertex which has no neighbors in $C$, then the objective value can be strictly improved by removing the vertex as a new community. Thus, $C$ consists of some pairs $x_i$ and $y_i$. Note that if $C$ contains two or more such pairs, then we similarly obtain a contradiction. Hence, we see that $C$ consists of only one pair $x_t$ and $y_t$.

From Lem.~\ref{lem: 1} and Lem.~\ref{lem: 2}, there exist communities $C_1, C_2, \dots, C_k$ corresponding to the bicliques $K_1, K_2, \dots, K_k$. In the following, these communities are termed \textit{biclique communities}. Assume that $C_\text{min}$ is one of those communities whose sum of the degrees is minimal. Now, the set of indices of the red element vertices in $C_\text{min}$ is denoted by $I_R$, and the same for the blue vertices is denoted by $I_B$. Then, the contribution of $C$ and $C_\text{min}$ to $Q_b$ is calculated by 
\begin{eqnarray*}
\begin{split}
&\frac{1}{m}\left(a^2 + \sum_{i \in I_R} a_i + \sum_{i \in I_B} a_i + |I_R \cap I_B| + 1\right) \\
&-\frac{1}{m^2}\left(\left((a+1)a + \sum_{i \in I_R}(ka_i + 2)\right) \left((a+1)a + \sum_{i \in I_B}(ka_i + 2)\right) + (ka_t + 2)^2 \right).
\end{split}
\end{eqnarray*}

Construct a new division $\mathcal{C}'$ by merging $C$ and $C_\text{min}$ into one community $C'$. The contribution of $C'$ to $Q_b$ is calculated by 
\begin{eqnarray*}
\begin{split}
&\frac{1}{m}\left(a^2 + \sum_{i \in I_R} a_i + \sum_{i \in I_B} a_i + |I_R \cap I_B| + 1 + 2a_t\right) \\
&-\frac{1}{m^2}\left((a+1)a + \sum_{i \in I_R}(ka_i + 2) + (ka_t + 2)\right) \left((a+1)a + \sum_{i \in I_B}(ka_i + 2) + (ka_t + 2) \right).
\end{split}
\end{eqnarray*}
Thus, we see that 
\begin{eqnarray*}
\Delta&:=&Q_b(\mathcal{C}') - Q_b(\mathcal{C}^*) \\
&=&\frac{2}{m}a_t - \frac{1}{m^2} \left(2(a+1)a(ka_t+2) + (ka_t+2)\left(\sum_{i \in I_R}(ka_i + 2) + \sum_{i \in I_B}(ka_i + 2)\right)\right).
\end{eqnarray*}
Now, we recall that $C_\text{min}$ is the biclique community whose sum of the degrees is minimal. Thus, the sum of the degrees of $C_\text{min}$ is less than or equal to the average of that of all biclique communities. Moreover, no biclique community contained element vertices $x_t$ and $y_t$. Therefore, it holds that  
\begin{equation*}
\sum_{i \in I_R}(ka_i + 2) + \sum_{i \in I_B}(ka_i + 2) \leq \frac{1}{k}(2ka + 12k - 2(ka_t+2)) < 2(a + 6).
\end{equation*}
Using these inequalities, we see that 
\begin{eqnarray*}
\Delta&>&\frac{2}{m^2} \left(ma_t - (ka_t + 2)(a^2 + 2a + 6) \right) \\
&\geq&\frac{2}{m^2} \left(\frac{6}{7}ka^2 + 3k - 2a^2 - 4a - 12 \right) \\
&\geq&\frac{2}{m^2} \left(\frac{4}{7}a^2 - 4a - 3 \right).
\end{eqnarray*}
The second inequality follows because $a_t \geq 1$. The last inequality follows because we can assume that $k > 3$ for all relevant instances of 3-PARTITION. Since $a \geq 3k > 9$, we obtain that $\Delta > 0$. This contradicts the optimality of division $\mathcal{C}^*$.

Next, we analyze the latter case. Assume that $C$ consists of red element vertices only. (Note that the following discussion is also applicable to every community which consists of blue element vertices only.) In this case, we assume that $C_\text{min}$ is one of the biclique communities whose sum of the degrees of the blue vertices is minimal. The set of indices of the red element vertices in $C_\text{min}$ is denoted by $I_R$, and the same for the blue vertices is denoted by $I_B$. Then, the contribution of $C$ and $C_\text{min}$ to $Q_b$ is calculated by 
\begin{eqnarray*}
\begin{split}
&\frac{1}{m}\left(a^2 + \sum_{i \in I_R} a_i + \sum_{i \in I_B} a_i + |I_R \cap I_B|\right) \\
&-\frac{1}{m^2}\left((a+1)a + \sum_{i \in I_R}(ka_i + 2)\right) \left((a+1)a + \sum_{i \in I_B}(ka_i + 2)\right).
\end{split}
\end{eqnarray*}

Construct a new division $\mathcal{C'}$ by transferring an arbitrary vertex $x_t$ from $C$ to $C_\text{min}$. If the corresponding element vertex $y_t$ is not contained in $C_\text{min}$, then the contribution of updated $C$ and $C_\text{min}$ to $Q_b$ is calculated by
\begin{eqnarray*}
\begin{split}
&\frac{1}{m}\left(a^2 + \sum_{i \in I_R} a_i + \sum_{i \in I_B} a_i + |I_R \cap I_B| + a_t \right) \\
&-\frac{1}{m^2}\left((a+1)a + \sum_{i \in I_R}(ka_i + 2) + (ka_t + 2)\right) \left((a+1)a + \sum_{i \in I_B}(ka_i + 2) \right).
\end{split}
\end{eqnarray*}
Note that if $y_t$ is contained in $C_\text{min}$, $1/m$ is added to the contribution. Thus, we see that 
\begin{eqnarray*}
\Delta&:=&Q_b(\mathcal{C}') - Q_b(\mathcal{C}^*) \\
&\geq&\frac{1}{m}a_t - \frac{1}{m^2} \left((a+1)a(ka_t+2) + (ka_t+2)\sum_{i \in I_B}(ka_i + 2)\right).
\end{eqnarray*}
Now, we recall that $C_\text{min}$ is the biclique community whose sum of the degrees of the blue vertices is minimal. Thus, the sum of the degrees of the blue vertices in $C_\text{min}$ is less than or equal to the average of that of all biclique communities. Moreover, no biclique community contained element vertex $x_t$. Therefore, it holds that 
\begin{equation*}
\sum_{i \in I_B}(ka_i + 2) \leq \frac{1}{k} \left(ka + 6k - (ka_t+2) \right) < a + 6.
\end{equation*}
Using these inequalities, we see that 
\begin{eqnarray*}
\Delta&>&\frac{1}{m^2} \left(ma_t - (ka_t + 2)(a^2 + 2a + 6) \right).
\end{eqnarray*}
In the analysis of the former case, we have already shown that $2\Delta > 0$. Thus, repeatedly applying this operation until $C = \emptyset$, we can obtain a new division with a strictly larger objective value. This contradicts the optimality of $\mathcal{C}^*$. 
\end{proof}

\begin{lem} \label{lem: 7}
In any division of $G(A)$ with maximal bipartite modularity, every pair of element vertices $x_i$ and $y_i$ belongs to the same community.
\end{lem}

\begin{proof}
Let us consider an arbitrary optimal division $\mathcal{C}^*$, and an arbitrary pair of element vertices $x_t$ and $y_t$. Suppose that $x_t$ belongs to a community $C_1$ and $y_t$ belongs to another community $C_2$. From Lem.~\ref{lem: 6}, we see that $C_1$ and $C_2$ are both biclique communities. The set of indices of the blue element vertices in $C_1$ and $C_2$ are denoted by $I_{B_1}$ and $I_{B_2}$, respectively. 

Construct a new division $\mathcal{C}'$ by transferring $x_t$ from $C_1$ to $C_2$. The increment of $Q_b$ is $1/m$ due to the edge between $x_t$ and $y_t$. On the other hand, the decrement of $Q_b$ is calculated by 
\begin{equation*}
\frac{1}{m^2} (ka_t + 2) \left(\sum_{i \in I_{B_2}}(ka_i + 2) - \sum_{i \in I_{B_1}}(ka_i + 2)\right),
\end{equation*}
because the degree of $x_t$ is $ka_t + 2$ and the sum of the degrees of the blue vertices in $C_2$ minus that of  $C_1$ is $\sum_{i \in I_{B_2}}(ka_i + 2) - \sum_{i \in I_{B_1}}(ka_i + 2)$. Thus, we see that 
\begin{eqnarray*}
\Delta &:=& Q_b(\mathcal{C}') - Q_b(\mathcal{C}^*) \\
&=&\frac{1}{m} - \frac{1}{m^2} (ka_t + 2) \left(\sum_{i \in I_{B_2}}(ka_i + 2) - \sum_{i \in I_{B_1}}(ka_i + 2)\right) \\
&\geq&\frac{1}{m} - \frac{1}{m^2} (ka_t + 2)(ka + 6k) \\
&>&\frac{1}{m} - \frac{1}{m^2} \left(\frac{a}{2} + 2\right) (ka + 6k) \\
&=&\frac{k}{m^2} \left(\frac{19}{14}a^2 - 3a - 3\right).
\end{eqnarray*}
The first inequality follows because the sum of the degrees of the blue vertices in the whole network is $ka + 6k$. The second inequality follows because we have $a_t < b/2 = a/2k$. Since $a \geq 3k \geq 3$, we immediately obtain that $\Delta > 0$. This contradicts the optimality of $\mathcal{C}^*$.
\end{proof}

So far, we have observed the conditions satisfied by divisions of $G(A)$ with maximal bipartite modularity. Finally, combining these findings, we present our result.

\begin{thm} \label{thm: 1}
BIMODULARITY is NP-complete in the strong sense.
\end{thm}

\begin{proof}
Since bipartite modularity for a given division $\mathcal{C}$ can be computed in polynomial time, BIMODULARITY belongs to the class NP. In what follows, we complete the reduction to show the NP-completeness. Recall that it is sufficient to provide appropriate parameter $K(A)$ such that $G(A)$ has a division $\mathcal{C}$ of $V$ which satisfies $Q_b(\mathcal{C}) \geq K(A)$ if and only if $A$ can be partitioned into $k$ sets with sum equal to $b$ each.

From the above observation, an arbitrary optimal division $\mathcal{C}^*$ of $G(A)$ in terms of  maximizing bipartite modularity can be represented as:
\begin{equation*}
\{C_1, C_2, \dots, C_k, X_1, X_2, \dots, X_{3k}, Y_1, Y_2, \dots, Y_{3k}\},
\end{equation*}
where $C_1, C_2, \dots, C_k$ are the biclique communities. We note that every pair of element vertices $x_i$ and $y_i$ belongs to one of the biclique communities. In this situation, the number of the edges within communities is unvarying. More specifically, denoting the number of such edges by $m_\text{intra}$, we have that
\begin{equation*}
m_\text{intra} = m - \left(2a(k-1) + 6k \right), 
\end{equation*}
because the number of the edges between different communities is always exactly $2a(k-1) + 6k$. Thus, we see that division $\mathcal{C}^*$ minimizes $\sum_{i=1}^k R_{C_i}B_{C_i}$. Since $R_{C_i} = B_{C_i}$ for $i=1,2,\dots,k$, it can be replaced by $\sum_{i=1}^k R_{C_i}^2$. Now, the sum of the degrees of the red vertices in all the biclique communities is given by  
\begin{equation*}
\sum_{i=1}^k R_{C_i} = ka(a + 1) + ka + 6k = k(a^2 + 2a + 6). 
\end{equation*}
Hence, the above sum of squares $\sum_{i=1}^k R_{C_i}^2$ has a lower bound: 
\begin{equation*}
k(a^2 + 2a + 6)^2.
\end{equation*}
This is attained if and only if all the sum of the degrees of the red vertices in each biclique community are the same, that is, for $i=1,2,\dots,k$, 
\begin{equation*}
R_{C_i} = \frac{1}{k} \sum_{i=1}^k R_{C_i} = a^2 + 2a + 6.
\end{equation*}

Now, assume that the lower bound is attained by division $\mathcal{C}^*$. Then, we see that the sum of the degrees of the red element vertices in each biclique community is equal to 
\begin{equation*}
a^2 + 2a + 6 - (a + 1)a = a + 6. 
\end{equation*} 
This implies that the number of the red element vertices in each biclique community is exactly three because we have $b/4 < a_i < b/2$ for $i=1,2,\dots,3k$. Thus, for each biclique community, three red element vertices, say $x_s$, $x_t$, and $x_u$, satisfy
\begin{equation*}
(ka_s + 2) + (ka_t + 2) + (ka_u + 2) = a + 6.
\end{equation*}
This leads that $a_s + a_t + a_u = a/k = b$. Therefore, $A$ of 3-PARTITION can be partitioned into $k$ sets with sum equal to $b$ each. 

Conversely, assume that $A$ can be partitioned into $k$ sets with sum equal to $b$ each. Then, we can assign three red element vertices, say $x_s$, $x_t$, and $x_u$, such that $a_s + a_t + a_u = b$ to each biclique community. It is easy to see that the lower bound can be attained by $G(A)$ constructed from such instance $A$. 

From the above, we should take $K(A)$ which is realized when the lower bound of the sum of squares is attained. Thus, we determine $K(A)$ as follows:
\begin{eqnarray*}
K(A) &=& \frac{m_\text{intra}}{m} - \frac{k(a^2 + 2a + 6)^2 + \frac{6}{7}ka^2 \left(\frac{1}{7}a^2 + 1\right)}{m^2} \\
&=&1-\frac{2a(k-1) + 6k}{m}- \frac{k(a^2 + 2a + 6)^2 + \frac{6}{7}ka^2 \left(\frac{1}{7}a^2 + 1\right)}{m^2}.
\end{eqnarray*}
This completes the desired reduction.
\end{proof}

\section{Concluding remarks} \label{sec: concluding_remarks}
In this study, we proved that maximizing bipartite modularity is NP-hard. This is the first computational complexity result for maximizing bipartite modularity. 

It is an interesting future direction to analyze the computational complexity beyond the NP-hardness. Indeed, as for the standard modularity, such computational complexity results have already been shown. For instance, Brandes \textit{et al.}~\cite{BrDeGaGoHoNiWa08} showed that maximizing the standard modularity remains NP-hard even when the number of the communities of an output division is restricted to exactly or at most two. In addition, DasGupta and Desai~\cite{DaDe13} showed that maximizing the standard modularity is APX-hard. This means that unless $\text{P} = \text{NP}$, there exists no polynomial-time approximation algorithm with approximation ratio $1 - \epsilon$ for some constant $\epsilon > 0$. 


\begin{acknowledgments}
The second author is supported by the Grant-in-Aid for JSPS Fellows. 
\end{acknowledgments}

\bibliography{bibfile}

\end{document}